\documentclass[aps,pra,10pt,showpacs,amsmath,amssymb,floatfix,superscriptaddress,longbibliography,twocolumn,showkeys]{revtex4-2}
\usepackage{physics}
\usepackage{graphicx}
\usepackage[dvipsnames,svgnames,table]{xcolor}
\usepackage{tcolorbox}
\usepackage{mathrsfs}
\usepackage{tikz}
\usepackage{tikz-3dplot}
\usetikzlibrary{quantikz2}
\usetikzlibrary{arrows.meta, bending, decorations.markings, angles, quotes, calc, 3d, decorations.pathreplacing}
\usepackage[unicode=true,
    pdfusetitle,
    bookmarks=true,
    bookmarksnumbered=false,
    bookmarksopen=false,
    breaklinks=true,
    pdfborder={0 0 0},
    backref=false,
    colorlinks=true,
hypertexnames=false]{hyperref}
\hypersetup{linkcolor=NavyBlue,urlcolor=NavyBlue,citecolor=NavyBlue}
\usepackage{amsthm}
\newtheorem{theorem}{Theorem}

\newcommand{\ii}{\mathrm{i}}
\newcommand{\ee}{\mathrm{e}}

\begin{document}
\title{Quantifying Uhlmann curvature from Yang-Mills-type action density and its implications in quantum multiparameter estimation}
\author{Yi-Lin Ge}
\author{Bing-Shu Hu}
\author{Ling-Yun Deng}
\email{denglingyun@hdu.edu.cn}
\affiliation{School of Sciences, Hangzhou Dianzi University, Hangzhou
310018, China}
\author{Xiao-Ming Lu}
\email{lxm@hdu.edu.cn}
\affiliation{School of Sciences, Hangzhou Dianzi University, Hangzhou
310018, China}
\affiliation{Zhejiang Key Laboratory of Quantum State Control and
    Optical Field Manipulation, Hangzhou Dianzi University, Hangzhou
310018, China}

\begin{abstract}
    The geometry of quantum states has profound implications in quantum multiparameter estimation.
    While the Riemannian structure of quantum state space is well understood, the full understanding of the curvature structure of mixed quantum states is still an open problem.
    Inspired by the Yang-Mills-type action density in non-Abelian gauge theory, we propose a scalar quantifying the Uhlmann curvature and establish its connection to the measurement incompatibility in quantum multiparameter estimation problems.
    We show that this curvature measure is gauge invariant, reparametrization invariant, and vanishes if and only if the Uhlmann curvature vanishes.
    We also explicitly calculate the Uhlmann curvature for the joint estimation of phase and phase diffusion as an example.
\end{abstract}


\maketitle

\section{Introduction}

Classical parameter estimation theory focuses on optimizing algorithms for estimating unknown parameters from observed data.
The precision of unbiased estimation in the asymptotic regime is determined by a Riemannian structure, known as the Fisher-Rao metric, on the statistical manifold~\cite{Fisher1922,Fisher1925,Rao1945}.
In the quantum regime, the quantum state space is endowed with an analogous Riemannian structure known as the quantum Fisher information (QFI) metric~\cite{Provost1980,Uhlmann1993,Braunstein1994}, which plays a crucial role in quantum parameter estimation theory~\cite{Helstrom1967,Helstrom1968,Helstrom1976,Paris2009,Liu2019}.
The idea behind such a geometric structure is that parameter estimation can be viewed as a process of distinguishing nearby quantum states on a manifold, and the QFI metric thereon quantifies the distinguishability of quantum states.

Quantum multiparameter estimation is much more complicated than its classical counterpart due to the non-commutativity of quantum observables.
Consequently, the quantum Cram\'er-Rao bound (QCRB) is not always attainable in the multiparameter regime when the optimal measurements for different parameters are incompatible~\cite{Yuen1973,Holevo2011}.
Matsumoto~\cite{Matsumoto1997,Matsumoto2005} conjectured that Uhlmann curvature reflects local properties of the quantum statistical model and proved that, for the faithful state model and pure state model, the QCRB is attainable if and only if the model is free of Uhlmann curvature.
Besides, the Uhlmann geometric phase, which is closely related to the Uhlmann curvature, has been shown to be a useful tool for characterizing phases transition~\cite{Viyuela2014,Silva2021,Hou2024,Wang2025}.
However, less is known about the relationship between the Uhlmann curvature and the measurement incompatibility for general quantum states.

In this work, we analyze the Uhlmann curvature from a geometric perspective and establish its connection to the measurement incompatibility in quantum multiparameter estimation problems.
Inspired by the Yang-Mills-type action density in non-Abelian gauge theory~\cite{Yang1954}, we propose a scalar quantifying the Uhlmann curvature and prove that it is gauge invariant, reparametrization invariant, and vanishes if and only if the Uhlmann curvature vanishes.
We show that the vanishing of this curvature measure is equivalent to the partial commutativity condition (PCC)~\cite{Yang2019b}, which is a necessary condition for the saturation of the QCRB for any state and sufficient when the states are pure or of full rank.
Furthermore, for two-parameter estimation problems with pure states, the curvature measure is proportional to the incompatibility factor in the tradeoff boundary of the estimation precision.

This paper is organized as follows.
In Sec.~\ref{sec:Uhlmann_gauge_theory}, we give a brief review of the Uhlmann connection and curvature for mixed quantum states.
In Sec.~\ref{sec:scalar_curvature}, we propose the scalar curvature measure and analyze its properties.
In Sec.~\ref{sec:quantum_multiparameter_estimation}, we study the implications of the Uhlmann curvature in quantum multiparameter estimation and show its relationship with the measurement incompatibility.
In Sec.~\ref{sec:example}, we explicitly calculate the Uhlmann curvature in the case of a two-parameter estimation problem.
We conclude in Sec.~\ref{sec:conclusion}.

\section{Uhlmann gauge theory}\label{sec:Uhlmann_gauge_theory}

Let us consider a manifold of density operators defined on the Hilbert space \(\mathcal{H}\).
We assume that the density operators have a fixed rank \(r\), ensuring the smoothness of the manifold.
Every density operator $\rho$ can be purified as a normalized vector $\ket{\Psi}$ in a dilated Hilbert space $\mathcal{H}\otimes \mathcal{H}'$ so that \(\tr_\mathrm{anc}\op{\Psi} = \rho\) with $\tr_\mathrm{anc}$ denoting partial trace over the auxiliary Hilbert space $\mathcal{H}'$.
The dimension of \(\mathcal H'\) must be at least the rank \(r\) of the density operators.
The purification states have a gauge degree of freedom.
If \(\ket{\Psi}\) is a purification state of \(\rho\),  so is \(\ket{\Psi'} = (I\otimes U)\ket{\Psi}\), where \(U\) can be any unitary operator acting on \(\mathcal{H}'\).
All unitary operators on \(\mathcal H'\) form the (local) gauge group of purification.

\begin{figure}
    \includegraphics{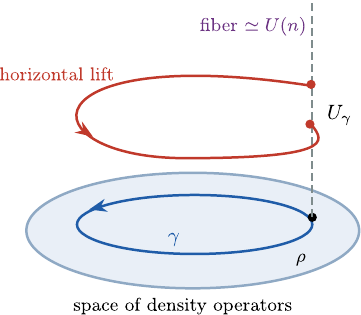}
    \caption{
        Schematic representation of the geometry underlying the Uhlmann connection.
        The density operators form a base manifold, and the purification states form a principal fiber bundle with the unitary group as the structure group.
        For a closed loop in the base manifold, the parallel transport of unitary operators along the loop results in a holonomy transformation \(U_\gamma = \hat P \exp(\oint_\gamma A)\) with \(\hat P\) denoting the path-ordering operator.
    }
    \label{fig:holonomy}
\end{figure}

Physical laws involving mixed states must be gauge invariant.
This can be achieved by establishing the ``connection''  structure on the purification bundle.
Let \(A\) be a connection 1-form~\cite{Frankel2011}, which takes values in the Lie algebra of the gauge group and transforms as
\begin{equation} \label{eq:connection_transform}
    A \to U A U^\dagger  + U \dd U^\dagger,
\end{equation}
under the gauge transformation \(\ket{\Psi} \to (I\otimes U)\ket{\Psi}\).
Note that \(A\) is anti-Hermitian, i.e., \(A^\dagger = - A\).
Given a connection, the covariant derivative \(\ket{D \Psi}\) on purification states is defined as
\begin{align} \label{eq:cov_deriv}
    \ket{D \Psi} := \ket{\dd\Psi} + (I \otimes A) \ket{\Psi},
\end{align}
which is gauge covariant, i.e., \(\ket{D \Psi} \to (I \otimes U) \ket{D \Psi}\).
With this covariant derivative, a gauge-invariant Riemannian metric \((\dd s)^2 = \ip{D \Psi}{D \Psi}\) on the manifold of density operators can be defined.
Besides, another important quantity is the curvature 2-form
\begin{align}
    F := \dd A + A\wedge A,
\end{align}
which is an operator acting on the ancillary Hilbert space $\mathcal{H}'$ and transforms covariantly under gauge transformations as \(F \to U F U^\dagger\).
The curvature 2-form \(F\) manifests the non-commutativity of the covariant derivative.
As shown in Fig.~\ref{fig:holonomy}, any infinitesimal closed loop \(\gamma\) in the base manifold of density operators induces a holonomy transformation \(U_\gamma \approx I + F \delta S \), where \(\delta S\) is the area of the infinitesimal loop.
In other words, a non-vanishing curvature 2-form \(F\) indicates that the unitary operators cannot be parallel transported along the loop in a consistent manner.

The Uhlmann connection~\cite{Uhlmann1986,Uhlmann1991Gauge,Uhlmann2011} is a special choice of \(A\) with which the Riemannian metric is exactly the Bures metric~\cite{Bures1969,Bengtsson2006}.
The Uhlmann connection \(A\) is determined by
\begin{equation} \label{eq:Uhlmann_connection}
    \ket{\dd\Psi}  = (G \otimes I)\ket{\Psi} - (I \otimes A) \ket{\Psi},
\end{equation}
where \(G\) is a Hermitian-operator-valued 1-form acting on \(\mathcal{H}\) and can be solved from the ansatz
\begin{align} \label{eq:rhoG}
    \dd\rho = G\rho + \rho G.
\end{align}
With the Uhlmann connection, the covariant derivative can be expressed as \(\ket{D \Psi} = (G \otimes I)\ket{\Psi}\) and results in the Bures metric~\cite{Bures1969,Bengtsson2006} \((\dd s)^2 = \tr(\rho G^2) = g_{\mu\nu} \dd x^\mu \dd x^\nu\), where
\begin{equation}
    g_{\mu\nu} = \tr(\rho \frac{G_\mu G_\nu + G_\nu G_\mu}{2})
\end{equation}
are the Bures metric tensor under the chosen parameterization of the density operators.
Note that the Bures metric is gauge invariant and thus independent of the choice of purification.

\section{Scalar curvature} \label{sec:scalar_curvature}

Inspired by the Yang-Mills-type action density  in non-Abelian gauge theory~\cite{Yang1954}, we propose a definition of scalar curvature as follows:
\begin{align}
    \mathcal{C}
    := - \frac14 \tr(F_{\mu\nu} F^{\mu\nu}),
\end{align}
where the indices are raised by the Bures metric as
\begin{align}
    F^{\mu\nu} = g^{\mu\alpha} g^{\nu\beta} F_{\alpha\beta},
\end{align}
with \(g^{\mu\nu}\) being the inverse of the Bures metric tensor \(g_{\mu\nu}\), i.e., \(g^{\mu\alpha} g_{\alpha\nu} = \delta^\mu_\nu\).
Note that the repeated indices are summed over according to the Einstein summation convention, and the Bures metric is assumed to be non-degenerate so that the inverse metric \(g^{\mu\nu}\) is well-defined.

The scalar curvature \(\mathcal{C}\) quantifies the local curvature information of the Uhlmann connection into a single scalar quantity.
The index contraction of \(F_{\mu\nu}\) with the inverse metric \(g^{\mu\nu}\) ensures that \(\mathcal{C}\) is invariant under reparametrization of the density operators.
The trace operation ensures that \(\mathcal{C}\) is gauge invariant.
Therefore, \(\mathcal{C}\) quantifies the curvature information of the Uhlmann connection at every point of the parameter space in an intrinsic manner.
The integral of the scalar curvature over the parameter space can be viewed as a global measure of the Uhlmann curvature, which is analogous to the Yang-Mills action in non-Abelian gauge theory.
Like the Yang–Mills action tells nature how much it ``costs'' to have a gauge field present, our scalar curvature $\mathcal{C}$ indicates ``the density of the cost'' of having noncommutativity among the parameters in the quantum model.

This scalar curvature \(\mathcal C\) has the following properties.
First, \(\mathcal{C}\) is gauge invariant, as \(F_{\mu\nu} \to U F_{\mu\nu} U^{-1}\) under any gauge transformation \(U\) and the trace of \(F_{\mu\nu} F^{\mu\nu}\) is invariant under such transformations.
Second, \(\mathcal{C}\) is invariant under coordinate transformation (i.e., reparametrization) of density operators, which is ensured by the contraction of \(F_{\mu\nu}\) with the inverse metric \(g^{\mu\nu}\).
Third, the curvature measure \(\mathcal{C}\) vanishes if and only if the curvature 2-form \(F\) vanishes.
This can be proved as follows.
Since \(\mathcal C\) is invariant under reparametrization, we can choose a local orthonormal coordinate system (by reparameterization) such that \(g_{\mu\nu}\) becomes the identity matrix and thus \(F^{\mu\nu} = F_{\mu\nu}\).
For such a case, we have \(-\frac14 \tr(F_{\mu\nu} F^{\mu\nu}) = \frac14 \tr(F_{\mu\nu}^\dagger F_{\mu\nu})\), which is non-negative and vanishes if and only if \(F_{\mu\nu}\) vanishes.
Noting that \(\mathcal{C}\) is a summation of non-negative terms, so it vanishes only if \(F_{\mu\nu} = 0\) for all pairs of parameters \((\mu,\nu)\), which is equivalent to the vanishing of \(F\).

While the above properties holds for any connection form, we now focus on the Uhlmann connection given by Eq.~\eqref{eq:Uhlmann_connection}.
We express the curvature measure \(\mathcal{C}\) in terms of the spectral decomposition of the density operator and the \(G_\mu\) operators.

\begin{theorem}
    Let \(\rho=\sum_{\ell=1}^r \lambda_\ell \op{\phi_\ell}\) be the spectral decomposition of \(\rho\) with nonzero eigenvalues \(\lambda_\ell\) and corresponding eigenvectors \(\ket{\phi_\ell}\).
    Let us define \(G^\mu = g^{\mu\alpha} G_\alpha\).
    The curvature measure \(\mathcal{C}\) can be expressed as
    \begin{align}
        \mathcal{C} &=
        \sum_{\ell,m=1}^r \frac{\lambda_\ell \lambda_m}{(\lambda_\ell + \lambda_m)^2}  \Theta_{\ell m},
        \label{eq:curvature_measure_spectral_decomposition}
    \end{align}
    with
    \begin{align} \label{eq:Theta}
        \Theta_{\ell m} := - \sum_{\mu\nu} \mel{\phi_\ell}{[G_\mu, G_\nu]}{\phi_m} \mel{\phi_m}{[G^\mu, G^\nu]}{\phi_\ell}.
    \end{align}
\end{theorem}

\begin{proof}
    We use the dual curvature 2-form \(\tilde{F} := \dd G - G \wedge G\) to derive Eq.~\eqref{eq:curvature_measure_spectral_decomposition}.
    For any purification state \(\ket{\Psi}\), the dual curvature \(\tilde{F}\)  satisfy the following relations:
    \begin{align}
        (I \otimes F)\ket{\Psi} &= (\tilde{F} \otimes I) \ket{\Psi}, \label{eq:F_dual_relation} \\
        \tilde{F}\rho + \rho \tilde{F}^\dagger &= 0. \label{eq:F_dual_rho_relation}
    \end{align}
    Equation~\eqref{eq:F_dual_relation} can be derived by taking the exterior derivative on both sides of Eq.~\eqref{eq:Uhlmann_connection} and using $\ket{\dd^2\Psi}=0$.
    Equation~\eqref{eq:F_dual_rho_relation} can be derived by taking the exterior derivative on both sides of Eq.~\eqref{eq:rhoG} and using \(\dd^2\rho=0\).

    In terms of the spectral decomposition, we choose the following gauge for the purification state:
    \begin{equation} \label{eq:Psi_spectral_decomposition}
        \ket{\Psi} = \sum_{n=1}^r \sqrt{\lambda_n} \ket{\phi_n} \otimes \ket{n},
    \end{equation}
    where \(\{\ket{n}\}\) is a coordinate-independent orthonormal basis of \(\mathbb{C}^r\).
    It follows from Eq.~\eqref{eq:F_dual_relation} and Eq.~\eqref{eq:Psi_spectral_decomposition} that
    \begin{align}
        \sum_{n=1}^r \sqrt{\lambda_n} \ket{\phi_n} \otimes F \ket{n}
        = \sum_{n=1}^r \sqrt{\lambda_n} \tilde{F} \ket{\phi_n} \otimes \ket{n}.
    \end{align}
    Left-multiplying the above equation by \(\bra{\phi_\ell} \otimes \bra{m}\) and using the orthonormality of the basis, we obtain
    \begin{align}
        \sqrt{\lambda_\ell}\mel{m}{F}{\ell}
        = \sqrt{\lambda_m} \mel{\phi_\ell}{\tilde{F}}{\phi_m}.
    \end{align}
    Therefore, we have
    \begin{align} \label{eq:F_dual_substitution}
        \tr(F_{\mu\nu} F^{\mu\nu}) &= \sum_{\ell,m=1}^r \mel{m}{F_{\mu\nu}}{\ell} \mel{\ell}{F^{\mu\nu}}{m} \nonumber \\
        &=\sum_{\ell,m=1}^r \mel{\phi_\ell}{\tilde{F}_{\mu\nu}}{\phi_m} \mel{\phi_m}{\tilde{F}^{\mu\nu}}{\phi_\ell},
    \end{align}
    where \( \tilde F_{\mu\nu} \) are given by
    \begin{align} \label{eq:F_dual_compenent}
        \tilde{F}_{\mu\nu} &= (\partial_\mu G_\nu - \partial_\nu G_\mu) - [G_\mu, G_\nu].
    \end{align}

    Notice that \(\dd G\) and \(G\wedge G\) are Hermitian and anti-Hermitian, respectively, and must satisfy Eq.~\eqref{eq:F_dual_rho_relation}.
    Left-multiplying the relation Eq.~\eqref{eq:F_dual_rho_relation} by \(\ket{\phi_m}\) and right-multiplying it by \(\ket{\phi_\ell}\), we obtain
    \begin{align}
        \mel{\phi_m}{\dd G}{\phi_\ell}
        =\frac{\lambda_\ell - \lambda_m}{\lambda_\ell + \lambda_m} \mel{\phi_m}{G \wedge G}{\phi_\ell},
    \end{align}
    implying that
    \begin{align}
        \mel{\phi_m}{\tilde{F}}{\phi_\ell}
        &= \mel{\phi_m}{(\dd G - G \wedge G)}{\phi_\ell} \nonumber\\
        &= - \frac{2 \lambda_m}{\lambda_\ell + \lambda_m} \mel{\phi_m}{G\wedge G}{\phi_\ell}.
        \label{eq:F_dual_matrix_element}
    \end{align}
    Substituting the above relation into Eq.~\eqref{eq:F_dual_substitution} leads to Eq.~\eqref{eq:curvature_measure_spectral_decomposition}.
\end{proof}

For pure states \(\rho=\op{\psi}\), the scalar curvature is reduced to
\begin{align} \label{eq:scalar_curvature_pure}
    \mathcal{C} = \sum_{\mu\nu} \Im\ev{G_\mu G_\nu}{\psi}  \Im\ev{G^\mu G^\nu}{\psi},
\end{align}
which is easy to compute from the imaginary part of the quantum geometric tensor~\cite{Berry1989}.

We also give an alternative definition of the curvature measure \(\mathcal{C}_\rho\) for which the trace operation is replaced by the expectation value on the purification state,
that is,
\begin{align}
    \mathcal{C}_\rho := - \frac14 \ev{I \otimes F_{\mu\nu} F^{\mu\nu}}{\Psi}.
\end{align}
Using the dual relation Eq.~\eqref{eq:F_dual_relation}, we have
\begin{align}
    \mathcal{C}_\rho = - \frac14 \tr(\rho \tilde{F}_{\mu\nu} \tilde{F}^{\mu\nu}).
\end{align}
This curvature measure \(\mathcal{C}_\rho\) is also gauge invariant and reparametrization invariant, and vanishes if and only if the Uhlmann curvature vanishes.
Moreover, with the spectral decomposition \(\rho=\sum_{\ell=1}^r \lambda_\ell \op{\phi_\ell}\), using Eq.~\eqref{eq:F_dual_compenent} and Eq.~\eqref{eq:F_dual_matrix_element}, we have
\begin{align}
    \mathcal{C}_\rho &= - \frac14 \sum_{\mu\nu} \sum_{\ell m} \lambda_\ell \mel{\phi_\ell}{\tilde{F}_{\mu\nu}}{\phi_m} \mel{\phi_m}{\tilde{F}^{\mu\nu}}{\phi_\ell} \notag \\
    &= \sum_{\ell m} \frac{\lambda_\ell^2 \lambda_m}{(\lambda_\ell + \lambda_m)^2} \Theta_{\ell m},
\end{align}
where \(\Theta_{\ell m}\) is defined in Eq.~\eqref{eq:Theta}.
As \(\Theta_{\ell m}\) is symmetric with respect to the indices \(\ell\) and \(m\), we can symmetrize the above expression and obtain
\begin{align} \label{eq:curvature_measure_spectral_decomposition2}
    \mathcal{C}_\rho = \sum_{\ell m} \frac{\lambda_\ell \lambda_m}{2(\lambda_\ell + \lambda_m)} \Theta_{\ell m}.
\end{align}
Comparing Eq.~\eqref{eq:curvature_measure_spectral_decomposition} and Eq.~\eqref{eq:curvature_measure_spectral_decomposition2}, we can see that \(\mathcal{C}\geq \mathcal{C}_\rho\).

There exists another relevant quantity called the mean Uhlmann curvature (MUC)~\cite{Carollo2018}, which is a real and anti-symmetric tensor defined as
\begin{align}
    \mathcal U_{\mu\nu} &:= \ii \ev{I \otimes F_{\mu\nu}}{\Psi} = \ii \tr(\rho \tilde{F}_{\mu\nu}) \notag \\
    &= -\ii \tr(\rho [G_\mu, G_\nu]).
\end{align}
Due to its simple definition, the MUC is relatively easy to compute and has been subsequently used in the field of quantum metrology~\cite{Carollo2019a,Wang2026a} and quantum phase transitions~\cite{Carollo2020}.
Based on the MUC, Carollo \emph{et al.}~\cite{Carollo2019a} proposed a scalar measure of the Uhlmann curvature as the amount of incompatibility in quantum multiparameter estimation, which is given by
\begin{equation}
    R  := \frac12 \norm{\ii\, g^{-1} \mathcal U}_\infty,
\end{equation}
where \(\norm{\bullet}_\infty\) denotes the maximal eigenvalue.
Although the quantity \(R\) is also gauge invariant and reparametrization invariant, it is not a faithful measure of the Uhlmann curvature since it can vanish even when the Uhlmann curvature does not vanish.
The is because the MUC is a linear functional of the Uhlmann curvature, and thus it may vanish even when the Uhlmann curvature, which is an operator, does not vanish.

Since the quantity \(R\) is reparametrization invariant, we can choose a local orthonormal coordinate system such that \(g\) becomes the identity matrix and thus \( R = \frac12 \norm{\ii\, \mathcal U}_\infty\).
For such a case, we have
\begin{align}
    R^2 = \frac14 \norm{\mathcal U^\mathrm{T} \mathcal U}_\infty
    \leq \frac14 \tr(\mathcal U^\mathrm{T} \mathcal U),
\end{align}
where \(\mathrm{T}\) denotes the matrix transpose.
This motivates us to define a new scalar measure of the Uhlmann curvature as
\begin{align}
    \mathcal{C}_\mathrm{MUC} := \frac14 \mathcal U_{\mu\nu} \mathcal U^{\mu\nu}
    = - \frac14 \tr(\rho \tilde{F}_{\mu\nu}) \tr(\rho \tilde{F}^{\mu\nu}).
\end{align}
With the local orthonormal coordinate system such that \(g\) becomes the identity matrix, we have \(\mathcal{C}_\mathrm{MUC} = \frac{1}{4} \sum_{\mu\nu} (\mathcal U_{\mu\nu})^2 \).
Using \(\mathcal U_{\mu\nu} = -\ii \sum_{\ell} \lambda_\ell \mel{\phi_\ell}{[G_\mu, G_\nu]}{\phi_\ell}\) and the Cauchy-Schwarz inequality \(|\sum_\ell \lambda_\ell x_\ell|^2 \leq (\sum_\ell \lambda_\ell) (\sum_\ell \lambda_\ell |x_\ell|^2)\) with \(\sum_\ell \lambda_\ell=1\), we obtain
\begin{align}
    \mathcal{C}_\mathrm{MUC} &\leq \frac14 \sum_{\mu\nu} \sum_{\ell} \lambda_\ell \abs{\mel{\phi_\ell}{[G_\mu, G_\nu]}{\phi_\ell}}^2 \notag \\
    & \leq \sum_{\ell m} \frac{\lambda_\ell \lambda_m}{2(\lambda_\ell + \lambda_m)} \Theta_{\ell m} = \mathcal{C}_\rho,
\end{align}
where the second inequality follows by including the non-negative off-diagonal terms (\(\ell \neq m\)) from Eq.~\eqref{eq:curvature_measure_spectral_decomposition2}.
Since \(R\), \(\mathcal{C}_\mathrm{MUC}\), \(\mathcal{C}_\rho\), and \(\mathcal{C}\) are all reparametrization invariant, we obtain
\begin{align}
    R^2 \leq \mathcal{C}_\mathrm{MUC} \leq \mathcal{C}_\rho \leq \mathcal{C},
\end{align}
for any parameterization of the density operators.

\section{Quantum multiparameter estimation}\label{sec:quantum_multiparameter_estimation}

Since the Uhlmann connection leads to the Bures metric, which is equivalent to the QFI matrix, it is natural to expect that the Uhlmann curvature plays a fundamental role in quantum multiparameter estimation~\cite{Matsumoto1997,Matsumoto2005}.
Matsumoto~\cite{Matsumoto1997} proved that, for the faithful state model, the QCRB is attainable if and only if the Uhlmann curvature vanishes.
We now generalize this result to the case of general quantum states as follows.

In the context of quantum multiparameter estimation, the density operators are parameterized by \(x^\mu\).
The unknown parameters are encoded in the quantum state and the goal is to estimate these unknown parameters with high precision through measurements on the quantum system.
The covariance matrix $\mathcal E$ of any unbiased estimator and any quantum measurement is bounded below by the QCRB~\cite{Helstrom1967,Helstrom1968}:
\begin{align}
    \mathcal E \geq n^{-1}\mathcal K^{-1},
    \label{eq:QCRB}
\end{align}
where \(n\) is the number of experiment repetitions and \(\mathcal{K}\) is the QFI matrix.
The elements of the QFI matrix are defined as
\begin{align}
    \mathcal K_{\mu\nu} = \tr (\rho \frac{L_\mu L_\nu + L_\nu L_\mu}{2}),
\end{align}
where the symmetric logarithmic derivative (SLD) operator $L_\mu$ is determined by $\partial_\mu \rho = \frac{1}{2} (L_\mu \rho + \rho L_\mu)$.
It can be seen from Eq.~\eqref{eq:rhoG} that the SLD operator is proportional to the component of \(G\) as \(L_\mu = 2 G_\mu\).
Therefore, the QFI matrix is equivalent to the Bures metric defined by the Uhlmann connection up to a factor of 4, i.e., \(\mathcal K_{\mu\nu} = 4 g_{\mu\nu}\).

\subsection{Saturation of the QCRB}

Unlike the single-parameter case~\cite{Braunstein1994}, the QCRB is not always saturable in the multiparameter regime~\cite{Yuen1973}.
The attainability of the bound requires the existence of a single measurement scheme that is optimal for estimating all parameters simultaneously.
Up to now, the most informative condition on the attainability of the QCRB is the partial commutativity condition (PCC)~\cite{Yang2019b} given by
\begin{align}\label{eq:PCC}
    \ev{[L_\mu, L_\nu]}{\psi} = 0 \quad \forall \ket{\psi} \in \operatorname{supp}\rho \qand \forall \mu, \nu,
\end{align}
where \(\operatorname{supp}\rho\) denotes the support of \(\rho\).
The PCC is necessary for the saturation of the QCRB for any state $\rho$~\cite{Yang2019b} and sufficient when the states are pure~\cite{Matsumoto2005} or of full rank~\cite{Matsumoto1997}.
We shall show that the PCC is equivalent to the vanishing of the Uhlmann curvature.

\begin{theorem}
    The PCC is satisfied if and only if \(\mathcal C=0\).
\end{theorem}

\begin{proof}
    Since \(\mathcal C\) is invariant under reparametrization, we can choose a local orthonormal coordinate system such that \(g_{\mu\nu}\) becomes the identity matrix and thus \(G^\mu = G_\mu\).
    It then follows from Eq.~\eqref{eq:curvature_measure_spectral_decomposition}  that \(\mathcal C=0\) if and only if \(\mel{\phi_\ell}{[G_\mu, G_\nu]}{\phi_m} = 0\) for all pairs of parameters \((\mu,\nu)\) and all pairs of eigenvectors \((\ket{\phi_\ell}, \ket{\phi_m})\) in the support of \(\rho\).
    Since \(L_\mu = 2 G_\mu\), the above condition is equivalent to \(\mel{\phi_\ell}{[L_\mu, L_\nu]}{\phi_m} = 0\) for all pairs of parameters \((\mu,\nu)\) and all pairs of eigenvectors \((\ket{\phi_\ell}, \ket{\phi_m})\) in the support of \(\rho\), which is exactly the PCC.
\end{proof}

\subsection{Tradeoff for two-parameter estimation}

Another important question in quantum multiparameter estimation is how to identify the tradeoff boundary of the estimation precision when the QCRB is not attainable.
The trade-off boundary is difficult to determine directly; it is often delineated by using various inequalities to exclude the regions that cannot be attained.
Within these inequalities, the incompatibility is often characterized by the commutators of the SLD operators, which are closely related to the Uhlmann curvature.

For two-parameter estimation problems, the covariance matrix \(\mathcal E\) of any unbiased estimator and any quantum measurement must satisfy the following inequality~\cite{Hu2025b}:
\begin{equation} \label{eq:tradeoff_boundary}
    \sqrt{|n\mathcal E \mathcal K - I|} + \sqrt{(1 - \gamma)|n\mathcal E\mathcal K|} \geq 1,
\end{equation}
where \(|\bullet|\) represents the determinant, \(I\) denotes the \(2\times2\) identity matrix, and the incompatibility factor   \(\gamma\) is defined as
\begin{align} \label{eq:incompatibility_facror}
    \gamma = \frac{ \norm{\sqrt{\rho}[L_1, L_2]\sqrt{\rho}}_1^2}{4\,|\mathcal K|},
\end{align}
with $\norm{X}_1 \equiv \tr\sqrt{X^{\dagger}X}$ denoting the Schatten-1 norm of an operator \(X\).
For pure states, the inequality~\eqref{eq:tradeoff_boundary} is tight in the asymptotical regime and thus the tradeoff boundary can be completely identified by the above inequality.
Meanwhile, the incompatibility factor \(\gamma\) can be simplified to
\begin{align} \label{eq:gamma_pure}
    \gamma = \frac{ \abs{\ev{[L_1, L_2]}{\psi}}^2 }{4\,|\mathcal K|}
    = \frac{ \abs{\Im\ev{G_1 G_2}{\psi}}^2 }{|g|}.
\end{align}

The incompatibility factor \(\gamma\) in Eq.~\eqref{eq:incompatibility_facror} quantifies the degree of incompatibility between the optimal measurements for estimating the two parameters in mixed states.
The numerator in Eq.~\eqref{eq:incompatibility_facror} originates from a strengthened version of the measurement uncertainty relation~\cite{Ozawa2014}, where the SLD operators are regarded as the ideal observables one intends to measure.
It is worth noting the following relation:
\begin{align}
    \norm{\sqrt{\rho}[L_1, L_2]\sqrt{\rho}}_1
    &= \max_U \abs{\tr(\sqrt{\rho}[L_1,L_2]\sqrt{\rho} U)} \notag \\
    &\geq \abs{\tr(\rho [L_1,L_2])},
\end{align}
where the last term is equivalent to the mean Uhlmann curvature~\cite{Carollo2018,Carollo2019a,Carollo2020,Wang2026a} up to a constant factor.
The denominator in Eq.~\eqref{eq:incompatibility_facror} is proportional to the determinant of the QFI matrix, which is related to the volume of the uncertainty ellipsoid in the parameter space.
The role of the denominator is twofold.
On one hand, the denominator \(4\,|\mathcal K|\) normalizes the incompatibility factor to the range \([0,1]\); on the other hand, it makes the incompatibility factor independent of the parameterization.

\begin{theorem}
    For two-parameter estimation problems with pure states, the curvature measure \(\mathcal{C}\) is proportional to the incompatibility factor \(\gamma\), i.e., \(\gamma = \mathcal{C}/2\).
\end{theorem}

\begin{proof}
    Since the scalar curvature \(\mathcal{C}\) is invariant under reparametrization, we can choose a local orthonormal coordinate system such that \(g_{\mu\nu}\) becomes the identity matrix.
    For such a case, we have \(G^\mu = G_\mu\) and \(|g|=1\).
    From Eq.~\eqref{eq:scalar_curvature_pure}, it can be shown that the scalar curvature is given by \(\mathcal{C} = 2 |\Im\ev{G_1 G_2}{\psi}|^2\), which is exactly \(2\gamma\) according to Eq.~\eqref{eq:gamma_pure}.
\end{proof}

\section{Example}\label{sec:example}

We now consider a concrete example of a two-parameter estimation problem to illustrate the quantification of the Uhlmann curvature.
The parameterized density matrix is given by
\begin{align}\label{eq:rho_th12_matrix}
    \rho=\frac{1}{2}
    \mqty(
        1 & \ee^{-\ii a - b } \\
        \ee^{\ii a - b} & 1
    )
\end{align}
with the standard basis, where \(a\) and \(b\) are two real parameters of interest, i.e., \(x^1=a\) and \(x^2=b\).
This model can be interpreted as the joint estimation of phase shift $a$ and phase diffusion $\sqrt{b}$~\cite{Vidrighin2014}.

After some algebra, we obtain the matrix representations of the \(G_\mu\) operators for $a$ and $b$ as
\begin{align}\label{eq:L1_L2}
    G_1 &= \frac12 \mqty(
        0 & -\ii \ee^{-b-\ii a} \\
        \ii \ee^{-b+\ii a} & 0
    ),\\
    G_2 &= \frac{1}{2(\ee^{2 b}-1)}\mqty(
        1 & -\ee^{b-\ii a} \\
        -\ee^{b+\ii a} & 1
    ).
\end{align}
The Bures matrix is then given by
\begin{align}
    g = \frac14 \mqty(
        \ee^{-2 b} & 0 \\
        0 & \frac{1}{\ee^{2 b}-1}
    ).
\end{align}
Substituting the explicit forms of $G_1$ and $G_2$ into Eq.~\eqref{eq:F_dual_compenent}, we obtain the dual curvature \(\tilde{F}_{12}\) as
\begin{align}
    \tilde{F}_{12}
    = \frac{-\ii}{2\left(\ee^{2 b}-1\right)}\left(
        \begin{array}{cc}
            1 & -\ee^{-b-\ii a} \\
            \ee^{-b+\ii a} & -1
    \end{array}\right).
\end{align}
As the Bures metric is diagonal, the curvature measure \(\mathcal{C}\) can be calculated by directly contracting the indices of \(\tilde{F}_{12}\) with the inverse metric \(g^{\mu\nu}\).
The resulting curvature measure is given by
\begin{align}
    \mathcal{C} = -\frac1{4 g_{11} g_{22}} \qty(
        \tr(\tilde{F}_{12}\tilde{F}_{12}) + \tr(\tilde{F}_{21}\tilde{F}_{21})
    ) =4.
\end{align}
This indicates that the manifold of density operators in this example is constantly curved and thus the QCRB cannot be saturated.
Since \(\mathcal C\) is invariant under reparametrization, the curvature measure \(\mathcal C=4\) remains unchanged under any smooth reparametrization of the parameters \(a\) and \(b\).
In other words, the constant curvature property of the manifold of density operators is an intrinsic geometric property of this model.
The state-weighted scalar measures of the Uhlmann curvature for this example is \(\mathcal{C}_\rho=2\).
Moreover, since the MUC tensor for this example is all-zero, the MUC-based scalar measures, namely \(R\) and \(\mathcal{C}_\mathrm{MUC}\), are zero and thus cannot capture the nontrivial curvature information of this model.

\section{Conclusion} \label{sec:conclusion}

In conclusion, we have proposed a scalar curvature measure to quantify the Uhlmann curvature for mixed quantum states and established its connection to the measurement incompatibility in quantum multiparameter estimation.
The proposed curvature measure is gauge invariant, reparametrization invariant, and vanishes if and only if the Uhlmann curvature vanishes.
We have shown that the vanishing of the curvature measure is equivalent to the PCC for the saturation of the QCRB in quantum multiparameter estimation problems.
For two-parameter estimation problems with pure states, we have shown that the curvature measure is proportional to the incompatibility factor in the tradeoff boundary of estimation precision.

We hope that our work can provide insights into the geometric structure of quantum state space and its implications in quantum multiparameter estimation.
As the scalar curvature is of the same form as the Yang-Mills action density in non-Abelian gauge theory, it would be interesting to explore the physical implications of such a curvature measure in quantum estimation theory and quantum information geometry.
Particullarly, it would be interesting to explore the prospects of using the scalar curvature in some rigorous lower bounds for an overall estimation error in multiparameter quantum estimation problems.
Meanwhile, the (operator-valued) Uhlmann curvature is conjectured to contain complete information about the tradeoff boundary of the estimation precision.
Lastly, we hope a deeper understanding of the Uhlmann curvature can provide new insights into the design of optimal measurement strategies that achieve the tradeoff boundary in quantum multiparameter estimation problems.

\begin{acknowledgments}
    This work is supported by the Quantum Science and Technology-National Science and Technology Major Project (Grant No.~2024ZD0301000), the National Natural Science Foundation of China (Grant Nos.~92476118 and 12275062), and Key R\&D Program of Zhejiang (Grant No. 2026C01004).
\end{acknowledgments}

\bibliography{reference}

\end{document}